\documentclass[12pt]{article}
\usepackage[dvipsnames]{xcolor}
\usepackage{etoolbox}
\usepackage{lmodern, geometry}
\usepackage{soul, wrapfig}
\usepackage{graphicx}
\usepackage{color}
\usepackage{amsmath}
\usepackage{amssymb}
\usepackage{amsfonts, csquotes}
\usepackage{latexsym, amsthm, mathrsfs, tabularx, enumerate}
\usepackage{hyperref}
\usepackage{fancybox}
\usepackage{amscd}
\usepackage{amsthm}
\usepackage{graphicx}
\usepackage[english]{babel}
\usepackage{tikz-qtree}
\usepackage{rotating}
\usepackage{stmaryrd}
\usepackage[round]{natbib}
\usetikzlibrary{arrows, automata, positioning}

\newtheorem{corollary}{Corollary}

\newtheorem*{definition}{Definition}

\newtheorem*{illustration}{Illustration}

\newtheorem{proposition}{Proposition}

\newtheorem{remark}{Remark}

\newcommand{\dom}{\textsf{dom}}

\newcommand{\ifif}{\Leftrightarrow}

\newcommand{\ran}{\textsf{ran}}

\newcommand{\RR}{\mathbb{R}}
\newcommand{\QQ}{\mathbb{Q}}
\newcommand{\N}{\mathbb{N}}

\newcommand{\Z}{\mathbb{Z}}

\input{tcilatex}

\begin{document}

\title{How rare are the properties of binary relations?\thanks{We sincerely thank Vicki Knoblauch for a careful reading of an earlier version of this manuscript.
Her constructive remarks and suggestions have helped us improve the content and exposition of this paper.}}
\author{Ram Sewak Dubey\thanks{%
Department of Economics, Feliciano School of Business, Montclair State
University, Montclair, NJ 07043, USA; E-mail: dubeyr@montclair.edu} \and 
Giorgio Laguzzi\thanks{%
University of Freiburg in the Mathematical Logic Group at Eckerstr. 1, 79104
Freiburg im Breisgau, Germany; Email: giorgio.laguzzi@libero.it}}

\date{\today}
\maketitle
\begin{abstract}
\citet{knoblauch2014} and \citet{knoblauch2015} investigate the relative size of the collection of binary relations with desirable features  as compared to the set of all binary relations using symmetric difference metric (Cantor) topology and Hausdorff metric topology.
We consider Ellentuck and doughnut topologies  to further this line of investigation.
We report the differences among the size of the useful binary relations in Cantor, Ellentuck and doughnut topologies.
It turns out that the doughnut topology admits binary relations with more general properties in contrast to the other two.
We further prove that among the induced Cantor and  Ellentuck topologies, the latter captures the relative size of partial orders among the collection of all quasi-orders.
Finally we show that the class of ethical binary relations is small in Ellentuck (and therefore in Cantor) topology but is not small in doughnut topology.
In essence, the Ellentuck topology fares better compared to Cantor topology in capturing the relative size of collections of binary relations.

\noindent \emph{Keywords:} \texttt{Cantor topology,}\; \texttt{Doughnut topology,}\;\texttt{Ellentuck topology,}\; \texttt{Ethical social welfare relation.} 

\noindent \emph{Journal of Economic Literature} Classification Numbers: 
\texttt{C65,} \texttt{D63,}\; \texttt{D71.}
\end{abstract}

\newpage

\section{Introduction}\label{s1}

This paper contributes to the research agenda introduced in recent papers, \citet{knoblauch2014} and \citet{knoblauch2015}, which deals with the study of some context free features of binary relations.
Given a binary relation, it would be interesting to know how likely it is that the binary relation is transitive or complete.
Usual investigation would rely on the context in which the binary relation has been defined.%
\footnote{For example, \citet{boudreau2018} analyze the likelihood of the preference known as rank-sum scoring violating transitivity using combinatorial techniques.} 
In contrast, \citet{knoblauch2014} explores such questions in abstract (context-free) settings.
Similar studies have been reported in \citet{klaska1997}, where the number of all possible transitive binary relation on finitely many alternatives have been determined.%
\footnote{The prevalence of context free individual preferences has also been used earlier in \citet{balasko1997} and \citet{tovey1997} to examine the super majority required to avoid Condorcet cycles.}

A brief review of relevant results is in order.
Following notation is standard in this literature.
Let $X= \{x_1, x_2, \cdots, x_n\}$ denote the set of all alternatives, having $n$ elements, i.e.,  $|X|=n$.
The Cartesian product $X\times X$ contains $n^2$ elements.
An arbitrary binary relation is an element of the power set of $X\times X$. 
Thus the number of all possible binary relations on $X\times X$ is $2^{n^2}$.
In this framework, \citet{klaska1997} considers every binary relations as equally likely and shows that the number of transitive binary relations tends to $2^{\frac{n^2}{4}}$.
Once can therefore infer that the transitive binary relations are relatively rare as a fraction of all binary relations as $n$ becomes infinitely large.
Enumeration of some of the common binary relations is listed in the Table \ref{T1} below.
Column (1) describes the property satisfied by the  binary relations  and column (2) quantifies all possible binary relation. 
The third column contains link to the related integer sequences available on the \href{https://oeis.org}{Online Encyclopedia of Integer Sequences},  which  is a very useful and freely available online resource.%
\footnote{In addition, the number of all possible asymmetric  and irreflexive binary relations are  $3^{\frac{n^2-n}{2}}$ and $2^{n^2-n}$ respectively. 
Note that there are other important class of binary relations  whose numerosity cannot be described as a  closed form function of the cardinality of the set of alternatives $n:=|X|$.
The collection of all possible partial orders, quasi-orders and equivalence classes belong to this category.
However, the estimates on the number of  all possible  partial orders, quasi-orders and equivalence classes  are available at the integer sequences \href{https://oeis.org/A001035}{A001035}, \href{https://oeis.org/A000798}{A000798} and  \href{https://oeis.org/A000110}{A000110} respectively.}

\begin{table}[ht]
\centering
\begin{tabular}{|c|c|c|}
\hline 
\textbf{Property}&	Number& \href{https://oeis.org}{OEIS}\\ \hline\hline
(1)&(2)&(3)\\ \hline
All binary relations& $2^{n^2}$ &\href{https://oeis.org/A002416}{A002416}\\\hline
Transitivity&  (asymptotic) $2^{\frac{n^2}{4}}$& \href{https://oeis.org/A006905}{A006905}\\\hline
Antisymmetry& $2^n3^{\frac{n^2-n}{2}}$&\href{https://oeis.org/A083667}{A083667}\\\hline
Symmetry& $2^{\frac{n^2+n}{2}}$&\href{https://oeis.org/A006125}{A006125}\\\hline
Reflexivity&	$2^{n^2-n}$&\href{https://oeis.org/A053763}{A053763}\\\hline
\end{tabular}
\caption{Number of binary relations satisfying basic properties in the finite $n$ element case} \label{T1}
\end{table}

In a recent paper, \citet{knoblauch2014} considers binary relations on $X$  containing countably infinitely many alternatives.
She examines the prevalence (or lack thereof) of  transitive (complete, asymmetric, antisymmetric, linear, etc.) binary relations among the collection of all binary relations.
It is easy to note that techniques applied in the finite case would no longer be informative in this situation since both  the transitive (complete, asymmetric, antisymmetric, linear, etc.) binary relations and all binary relations turn out to be infinite with their ratio possibly indeterminate.
\citet{knoblauch2014} devises an alternative measure for this investigation.
She defines a symmetric difference metric topology on the collection of all binary relations on the set  $X\times X$.%
\footnote{Symmetric difference metric on the collection all of binary relations on a \emph{finite} set was introduced by \citet{kemeny1962}.}
She shows that the collection of transitive (or asymmetric, antisymmetric, complete, linear etc.) relations has a dense open complement in the symmetric difference metric topology.%
\footnote{In the social choice literature, metric on preference profiles on infinite set of alternatives have been defined earlier in \citet{chichilnisky1980}.
She considers a metric on \emph{smooth} preference profiles on a manifold to study the social choices.
We note that the metric spaces considered in  \citet{knoblauch2014} and \citet{knoblauch2015} do not require the preference profiles to be smooth.
In this sense, the setting is more general compared to \citet{chichilnisky1980}.
We adopt the general setting for our investigation.}

In this sense  the collection of binary relations having desirable properties are \emph{rare} among the set of all binary relations.
Further, the rarity feature continues to hold when the nowhere dense sets with respect to symmetric difference metric topology are replaced by the notion of measure zero set with respect to Lebesgue measure.
In \citet{knoblauch2015} the analysis is  extended to the case of uncountably infinite $X$ and to the Hausdorff metric topology.%
\footnote{The Hausdorff metric measures the distance between two binary relations is the largest distance from an ordered pair in one binary relation to the other binary relation. 
In contrast, in the symmetric difference metric, the distance between two binary relations, $A$ and $B$, is the number of disagreements, i.e., the size of the disagreement set $A\setminus B$ and $B\setminus A$.}
The choice of symmetric difference metric and Hausdorff metric enables her to show that the rarity of any property of the binary relation depends on the topology at hand.
In particular, a collection%
\footnote{of equivalence classes of Lebesgue measurable binary relations on $[0,1]$}
containing asymmetric (antisymmetric, transitive) binary relations is shown to be rare in the symmetric difference metric topology but not in the Hausdorff metric topology.
The results in \citet{knoblauch2015} also shed light on the properties of binary relations which are rare in both topologies.%
\footnote{The author describes it as  there is \emph{agreement} between the two topologies.
The situation when the rarity property holds in one but not the other is termed as \emph{disagreement} between the two topologies.}

In this paper we consider the symmetric difference metric topology on the set $S:= \{0,1\}^{\N}$ from \citet{knoblauch2014} and call it the \emph{Cantor topology}.
We assume the set of alternatives $X$ to be countably infinite which implies that the cardinality of the set $X\times X$ is also countably infinite. 
This allows us to describe any arbitrary binary relation via a coding set which is an element of $S:= \{0,1\}^{\N}=2^\N$.
This slight variation of the definition is useful to then approach two other topological concepts we are going to deal with.%
\footnote{It is well-known and straightforward to check that Cantor topology and symmetric different topology coincide.}  
The result (based on symmetric difference metric topology) in \citet{knoblauch2014} could be re-phrased as follows.
The coding set of binary relations satisfying some well-known \emph{basic} properties (e.g., transitivity, completeness, asymmetry, antisymmetry, linearity) are nowhere dense subsets of $2^\N$ with respect to Cantor topology and they are null sets with respect to Lebesgue measure.

Relying on the coding set, we define Ellentuck and doughnut topology to further explore the rarity of the binary relations satisfying the basic properties.
Among the three topologies, i.e., Cantor, Ellentuck and doughnut, the doughnut topology is the finest followed by Ellentuck and Cantor which is the coarsest.
Though the transitive, asymmetric or antisymmetric binary relations have been proven to be rare in Cantor topology in \citet{knoblauch2014}, we show in Proposition \ref{P1} that they are not rare in the Ellentuck topology.
Proposition \ref{P1} also demonstrates that the irreflexive binary relations are not rare in Ellentuck topology.
We next examine complete, reflexive or symmetric binary relations, which are rare in Cantor topology as established in \citet{knoblauch2014}.
We present two results in Proposition \ref{P2} on these binary relations.
It turns out that complete, reflexive or symmetric binary relations are rare in Ellentuck topology.
However they are not rare in the doughnut topology, the finest one considered in this paper.
Corollary \ref{C1} based on  the two propositions proves that the linear binary relations are not rare in doughnut topology while they are rare in Cantor topology (\citet{knoblauch2014}) and Ellentuck topology (follows from Proposition \ref{P2} in this paper).

We next examine the relative size of the collection of all quasi-orders in the collection of all transitive binary relations. 
Asymptotic results from the finite $X$ case show that the fraction of all transitive orders which are  quasi-orders  tends to zero.
In Proposition \ref{P3}, we show that the collection of quasi-orders are rare in the relative Ellentuck topology induced by the collection of all transitive binary relations.
Combinatorial results from the finite $X$ case also verify that the share of quasi-orders which are partial orders approaches one in the limit.
We present two results in the topological setting.
In Proposition \ref{P4} we prove that  the partial orders are rare among the collection of quasi-orders in the induced Cantor topology.
However, the partial orders are not rare among the collection of quasi-orders in the induced Ellentuck topology as we show in Proposition \ref{P5}.
In this framework, Ellentuck topology seems to be fine enough to capture the intuition from the asymptotic analysis to the abstract infinite setting whereas the Cantor topology is quite coarse for such investigations.

The second class of binary relations examined in this paper are those satisfying equity or efficiency properties from the social choice literature. 
The equity principles can be classified in two broad categories, procedural and consequentialist.
A brief description of the anonymity axiom (a version of procedural equity) and the strong equity axiom (an example of consequentialist equity) is as follows.

\citet{ramsey1928} while devising the idea of inter-generational equity, observed  that discounting one generation’s utility relative to another's is \enquote{ethically indefensible}, and something that \enquote{arises merely from the weakness of the imagination}. 
\citet{diamond1965} formalized the concept of \enquote{equal treatment} of all generations (present and future) in the form of an \emph{anonymity} axiom on social preferences.
Anonymity requires that society should be indifferent between two infinite streams of well-being, if one is obtained from the other by interchanging the levels of well-being of any two generations.
It is an example of procedural equity principle; i.e., the change involved in the pair of utility streams being ranked does not alter the distribution of utilities.

Equity principles which require an alteration in the distribution of utilities are called consequentialist equity concepts.
The equity notions examined in this paper is the \emph{strong} equity (SE) axiom.%
\footnote{Originally the term \enquote{strong equity} has been used relative to the axiom introduced by \citet{hammond1976}, which he called the \enquote{Equity Axiom}.
Hammond explains the axiom to be in the spirit of the \enquote{Weak Equity Axiom} of \citet{sen1973}
Strong equity was introduced by \citet{aspremont1977}, who referred to it as an \enquote{Extremist Equity Axiom}.}  
Strong Equity compares two infinite utility streams ($x$ and $y$) in which all generations except two have the same utility levels in both utility streams. 
Regarding the two remaining generations (say, $i$ and $j$), one of the generations (say $i$) is better off in utility stream $x$, and the other generation ($j$) is better off in utility stream $y$, thereby leading to a situation of a conflict. 
Strong equity requires that if for both utility streams, it is generation $i$ which is worse off than generation $j$, then generation $i$ should choose (on behalf of the society) between $x$ and $y$. 
Thus, utility stream $x$ is socially preferred to $y$, since generation $i$ is better off in $x$ than in $y$.

The efficiency principle considered here is based on a broad consensus among scholars as a desirable attribute that social binary relations preferences should possess, namely the Pareto criteria. 
The strongest version of the Pareto principle asserts that one utility stream must be deemed strictly preferred compared to another if at least one generation is better off and no generation is worse off in the former compared to the latter. 

We retain the cardinality of $X\times X$ as countably infinite and show in Propositions \ref{P6}-\ref{P8} that the binary relations satisfying each of these axioms are not rare in doughnut topology whereas they are rare in Ellentuck (and similarly in Cantor) topology.

To summarize, we have used the Ellentuck and Doughnut small sets to provide interesting insights on the notion of rarity of  the class of binary relations endowed with basic features or desirable equity/efficiency axioms when  the cardinality of the set $X\times X$ is countably infinite.

Remainder of the paper is organized as follows.
We introduce notation and the definitions in section \ref{s2}.
The binary relations satisfying basic properties are examined in section \ref{s3} and \ref{s4}.
Section \ref{s5} deals with the binary relations satisfying equity or efficiency axioms.
We conclude in section \ref{s6}.

\section{Preliminaries}\label{s2}
\subsection{Notation}\label{s21}
Let $\RR$, $\RR_+$ $\N$, $\QQ$, $\Z$ be the sets of real numbers, non-negative real numbers, natural numbers, rational numbers and integers respectively. 
For all $y, z\,\in \RR^{\N}\,$, we write $y\geq z$ if $y_{n}\geq z_{n}$, for all $n\in \N$; we write $y>z$ if $y\geq z$ and $y\neq z;$ and we write $y\gg z$ if $y_{n}>z_{n}$ for all $n\in \N$.

\subsection{Binary relations}
A \emph{binary relation} on a set $X$ is a subset of $X\times X$.
Preference over set $X$ could be considered as a binary relation $\Re$ on $X$ as follows.
$x$ is preferred to $y$ if and only if $(x, y)\in \Re$.
A binary relation $\Re$ on a set $X$ is 
\begin{itemize}
\item[]{\emph{reflexive} if $(x, x)\in \Re$, for all $x \in X$,}
\item[]{\emph{irreflexive} if $(x, x)\notin \Re$, for all $x \in X$,}
\item[]{\emph{complete} if $(x, y)\in \Re$ or $(y, x)\in \Re$, for all $x, y\in X$,}
\item[]{\emph{transitive} if $(x, y), (y, z)\in \Re$ implies $(x, z)\in \Re$, for all $x, y, z\in X$,}
\item[]{\emph{symmetric} if $(x, y)\in \Re$ implies $(y, x)\in \Re$, for all $x, y\in X$,}
\item[]{\emph{asymmetric} if $(x, y)\in \Re$ implies $(y, x)\notin \Re$, for all $x, y\in X$,}
\item[]{\emph{antisymmetric} if $(x, y), (y, x)\in \Re$ implies $x=y$, for all $x, y\in X$,}
\item[]{\emph{quasi-order} if $\Re$ is transitive and reflexive,}
\item[]{\emph{partial order} if $\Re$ is transitive, reflexive and antisymmetric,}
\item[]{\emph{equivalence relation} if $\Re$ is transitive, reflexive and symmetric,}
\item[]{\emph{linear order} if $\Re$ is complete, transitive and antisymmetric.}
\end{itemize}
\subsection{Coding set}
Let $X$ be a countably infinite set of alternatives (utility streams) and let $\{x_n: n \in \N\}$ be the  enumeration of  all elements in $X$.
Let $\{q_k: k \in \N \}$ enumerate all pairs in $X \times X$. 
A binary relations $\Re$ on $X$ can be then coded/seen as a subset of $\N$ by collecting those indices $k \in \N$ for which the corresponding pair $q_k \in \Re$. 
Moreover, by standard identification of a subset with its characteristic function one can think of  $\{0, 1\}^{\N}$ (also denoted by $2^\N$)%
\footnote{We use both the notations interchangeably throughout the paper.} as the set of all codes of binary relations on $X$. 
We denote the binary sequence in $2^\N$ coding the binary relation $\Re \subseteq X \times X$ by $z_{\Re}$. 
The notation $\Re_z$ is used for the binary relation coded by $z \in 2^\N$.
Following examples would be helpful to clarify the notations.

\begin{illustration}
\emph{\textbf{Finite case}:
Let the set $X$ containing all alternatives be finite, with the cardinality $|X|=n$.
Then the number of all possible binary relations is $2^{\left(n^2\right)}$.
Consider $X:=\{x_1, x_2, x_3\}$, i.e., $n=3$. 
Enumeration of all pairs in $X\times X$ is $q_1=(x_1, x_1)$, $q_2=(x_1, x_2)$, $\cdots$, $q_9=(x_3, x_3)$.
There are $2^{3^2} = 2^9$ possible binary relations.
A binary relation $\Re$ could be the collection $\{q_1, q_5, q_9\}$ (the reflexive binary relation).
Then $z_{\Re} = \{1, 0, 0, 0, 1, 0, 0, 0, 1\}$.
Also $\Re^{\prime}_{z^{\prime}}$ for $z^{\prime}=\{1, 1, 1, 0, 0, 0, 0, 0, 0\}$ is the binary relation $\Re^{\prime}= \{(x_1, x_1), (x_1, x_2), (x_1, x_3)\}=\{q_1, q_2, q_3\}$.}

\emph{\textbf{Countably infinite case}:
Let the set $X= \{x_1, x_2, x_3, \cdots\}$ contain countably infinitely many elements.
The number of all possible pairs of elements in set $X\times X$ is countably infinite as well, with the enumeration  denoted by $q_k$, $k\in \N$.
We fix this enumeration of the pairs and define binary relations as a sequence $z\in \{0,1\}^{\N}$.
Thus, $z= \{1, 1, \cdots\}$ describes a binary relation containing all pairs of alternatives, i.e., $\Re_{z} = \{(x_i, x_j): x_i, x_j\in X\; \forall i, j\in \N\}$.
If $q_1=\left(x_1, x_3\right)$, then $z^{\prime}= \{0, 1, 1, \cdots\}$ describes a binary relation containing all pairs of alternatives except $\left(x_1, x_3\right)$ i.e., $\Re_{z^{\prime}} = \Re_{z} \setminus \left(x_1, x_3\right)$. }
\end{illustration}

Let $p$ be a property of binary relation.
We define $C(p)$ as the collection of all codes of binary relations satisfying property $p$, i.e. $C(p):= \{z_\Re \in 2^\N: \Re\; \text{satisfies}\; p \}$. 
Thus, $C(p)$ is the \emph{coding set} of the property $p$. 
For instance, if $p$ is the property \emph{transitivity}, then $C(p)$ is the subset of $2^\N$ collecting all codes of all transitive binary relations.
Ideally, one expects that the useful properties (transitivity, asymmetry, etc.) to be ubiquitous. 
However, recent papers, (see \citet{knoblauch2014}, \citet{knoblauch2015})  have shown that these properties turn out to be quite rare in appealing topologies.
In order to formalize this kind of questions is to evaluate how \emph{small} is the corresponding coding set $C(p)$. 
\subsection{Topological space}
A \emph{topological space} $(X, \tau)$ is a set $X$ together with a topology $\tau$, a collection of subsets of $X$ such that 
\begin{itemize}
\item{$X\in \tau$, and $\emptyset\notin \tau$,}
\item{every union of members of $\tau$ is a member of $\tau$,}
\item{every finite intersection of members of $\tau$ is a member of $\tau$.}
\end{itemize}
The members of $\tau$ are called the \emph{open sets}, and their complements are called the \emph{closed sets}.
A \emph{basis} $B$ for the topological space $(X, \tau)$ is a subcollection of $\tau$ such that $p\in O\in \tau$ implies $p\in b\subseteq O$ for some $b\in B$.
A set $S\subseteq X$ is \emph{dense} if $O\in \tau$ and $O\neq \emptyset$ together imply $O\cap S\neq \emptyset$.
A set $A \subseteq X$ is $\tau$-\emph{nowhere dense} if and only if the interior of the closure of $A$ is empty.
Given the topological space $(X\times X, \tau)$ and a subset $S\subseteq X\times X$, we define the induced (or relative) topology, $\tau_S$ as follows.
\[
\tau_S = \{S\cap B: B\in \tau\}.
\]

\subsection{Ideal}

The intuitive idea of relative size of a collection of subsets with respect to the entire collection of subsets (the power set $\mathcal{P}(S):=\{A: A \subseteq S\}$) of any given set $S$ is via the notion of  \emph{ideal}.
Elements of ideal are considered \emph{small} subsets of $S$.
Definition of ideal is as follows.
\begin{definition}
\emph{Ideal: 
Given a set $S$ and $I \subseteq \mathcal{P}(S)$, we say that $I$ is an ideal if and only if  the following three properties hold.
\begin{itemize}
\item $\emptyset \in I$ and $S \notin I$,
\item $\forall A, B \in I$, $A \cup B \in I$, and
\item $\forall A \in I\; \forall B \in \mathcal{P}(S)$, if $B \subseteq A$ then $B \in I$.
\end{itemize}}
\end{definition}

For instance, if $S$ is endowed with a topology $\tau$, a well-established notion is the ideal of $\tau$-nowhere dense subsets of $S$. 
Our objective is to analyze various notions of \emph{smallness} of a collection of subsets of a set so as to extend the investigation on the rarity of  properties of binary relations. 
We employ two notions of smallness (borrowed from the descriptive set theory), namely the ideal of \emph{Ramsey} null sets (also called \emph{Ellentuck} nowhere dense sets) and the ideal of \emph{doughnut} null sets.
Following definitions are needed in order to capture these notions.

\begin{definition}\emph{Ideal $\mathcal{I}_U$ of \emph{$U$-small subsets} of $2^\N$:
Let $U$ be a non-empty collection of subsets of $2^\N$ (i.e. $U \subseteq \mathcal{P}(2^\N)$) such that:
\begin{itemize}
\item for all $u \in U$ there exists non-empty $u^{\prime} \subseteq u$ such that  $u^{\prime} \in U$, and
\item for all $x \in 2^\N$ there exists $u \in U$ such that $x \in u$. 
\end{itemize}
Set $X \in \mathcal{I}_U$ if and only if for every $u \in U$ there exists  non-empty $u^{\prime} \in U$, $u^{\prime} \subseteq u$ such that $u^{\prime} \cap X = \emptyset$. }
\end{definition}

\begin{remark}
\emph{Note that the notion of $U$-small subsets generalizes the concepts of nowhere dense and Lebesgue null sets. 
Indeed if $U$ is the collection of all open sets with respect to the Cantor topology, then $\mathcal{I}_U$ is exactly the ideal of nowhere dense sets with respect to the  Cantor topology.
It is easy to check that $\mathcal{I}_U$ is the ideal of Lebesgue measure zero sets when we consider $U$ to be the collection of all closed subsets of $2^\N$ with positive Lebesgue measure.}
\end{remark}

\subsection{Cantor, Ellentuck and doughnut collections}

A \emph{partial function} $f: X \rightarrow Y$ is a function from a subset $S$ of $X$ to $Y$. 
If $S$ equals $X$, the partial function is said to be total.
Domain and range of function $f$ are denoted by $\dom(f)$ and $\ran(f)$ respectively.
\begin{definition}
Let $f: \N \rightarrow \{0, 1\}$ be a partial function and define $N_f:=\{x \in 2^\N: \forall n \in \dom(f)(x(n)=f(n))\}$.
\begin{itemize}
\item[] \emph{Cantor} collection $\gamma$: It  consists of all $N_f$ such that $\dom(f)$ is finite.
\item[] \emph{Ellentuck} collection $\varepsilon$: It  consists of all $N_f$ such that $\dom(f)$ and $\N \setminus\dom(f)$ are both infinite and there exists $k \in \N$ for all $n \in \dom(f) (n \geq k \Rightarrow f(n)=0)$.%
\footnote{See \citet[p. 524, Definition 25.26]{jech2003} for a textbook definition of the Ellentuck topology. 
\citet{brendle2005} defines and investigates the properties of doughnut topology.} 
\item[] \emph{Doughnut} collection $\delta$: It  consists of all $N_f$ such that $\dom(f)$ and $\N \setminus \dom(f)$ are both infinite.
\end{itemize}
\end{definition}

It is straightforward to note that these three collections are nested with $\delta$ the finest, $\gamma$ the coarsest and $\varepsilon$ lying in the middle, i.e., $\gamma \subseteq \varepsilon \subseteq \delta$. 
Further, the Cantor topology is metrizable (i.e., it is first countable) whereas both Ellentuck and doughnut topologies are non-metrizable.
In the literature, the sets in $\mathcal{I}_\varepsilon$ are also called \emph{Ramsey} null (or \emph{Ellentuck} nowhere dense) and the sets in $\mathcal{I}_\delta$ are also called \emph{Doughnut} null.

\subsection{Equity and Pareto principles}
We will be dealing with following equity and Pareto principles in this paper.
The anonymity (also called finite anonymity) axiom is a notion of procedural equity.
Strong equity belongs to the class of consequentialist equity.
The efficiency notion we use is the standard Pareto principle. 
Definitions and formal notations are as below.
\begin{definition} 
Let $\Re \subseteq X \times X$ be a binary relation.  
\begin{itemize}
\item[]\emph{Anonymity:} $\Re$ is said to be \emph{anonymous} if and only if for all $\mathbf{t}, \mathbf{t}^{\prime} \in X=Y^N$ there are $i, j \leq N$, $\mathbf{t}(j) = \mathbf{t}^{\prime}(i)$ and $\mathbf{t}(i)=\mathbf{t}^{\prime}(j)$ and for all $k \neq i,j$, $\mathbf{t}(k)=\mathbf{t}^{\prime}(k)$, then  $(\mathbf{t}^{\prime}, \mathbf{t})\in \Re$ and $(\mathbf{t},  \mathbf{t}^{\prime})\in \Re$ hold, i.e., ($\mathbf{t}\sim_a \mathbf{t}^{\prime}$).
\[
\mathbf{t} \sim_a \mathbf{t}^{\prime} \ifif \exists i, j \leq N ( \mathbf{t}(j) = \mathbf{t}^{\prime}(i) \land \mathbf{t}(i)= \mathbf{t}^{\prime}(j) \land \forall k \neq i, j (\mathbf{t}(k) =\mathbf{t}^{\prime}(k))). 
\]
\item[]\emph{Strong equity:} $\Re$ is said to satisfy  \emph{strong equity} if and only if for all $\mathbf{t}, \mathbf{t}^{\prime} \in Y^N$ there exist $i, j \leq N$ such that $\mathbf{t}(i) < \mathbf{t}^{\prime}(i) < \mathbf{t}^{\prime}(j) < \mathbf{t}(j)$ and for all $k \neq i,j$, $\mathbf{t}(k)=\mathbf{t}^{\prime}(k)$, then  $(\mathbf{t}^{\prime}, \mathbf{t})\in \Re$ and $(\mathbf{t},  \mathbf{t}^{\prime})\notin \Re$, i.e., ($\mathbf{t}<_s \mathbf{t}^{\prime}$).
\[
\mathbf{t} <_s \mathbf{t}^{\prime}\ifif \exists i,j \leq N (\mathbf{t}(i) < \mathbf{t}^{\prime}(i) < \mathbf{t}^{\prime}(j) < \mathbf{t}(j)) \land \forall k \neq i,j (\mathbf{t}(k)=\mathbf{t}^{\prime}(k)).
\]
\item[]\emph{Pareto principle:} $\Re$ is said to be \emph{Paretian} if and only if for all $\mathbf{t}, \mathbf{t}^{\prime} \in Y^N$ for all  $i \leq N$, $\mathbf{t}(i) \leq \mathbf{t}^{\prime}(i)$ and there exists $i \leq N$ such that $\mathbf{t}(i) < \mathbf{t}^{\prime}(i)$, then $(\mathbf{t}^{\prime}, \mathbf{t})\in \Re$ and $(\mathbf{t},  \mathbf{t}^{\prime})\notin \Re$, i.e., ($\mathbf{t}<_p \mathbf{t}^{\prime}$)
\[
\mathbf{t} <_p \mathbf{t}^{\prime} \ifif \forall i \leq N (\mathbf{t}(i) \leq \mathbf{t}^{\prime}(i)) \land \exists i \leq N (\mathbf{t}(i) < \mathbf{t}^{\prime}(i)). 
\]
\end{itemize}
\end{definition}

\section{Basic properties of binary relations} \label{s3}

In economic theory social welfare relations and preference relations satisfy some properties, which we may a
priori split into two categories: \emph{basic} properties, such as reflexivity, irreflexivity, symmetry, asymmetry and transitivity; \emph{economic} property, such as Pareto, anonymity, strong equity principles. 
In this section we consider the basic properties of binary relations.
Following partition of $\N$ would be useful in the proofs.
Let 
\begin{equation}\label{E1}
R:= \{n_k \in \N: \;\text{there exists}\; k\in\N, \;\text{such that}\; q_{n_k}= (x_k, x_k)\}.
\end{equation}
Set $R$ is both infinite and co-infinite and collects the enumeration of all pairs in $X\times X$ with identical elements chosen in the pair (i.e., the reflexive part of any binary relation).
Note that for all $n\in \N\setminus R$, $q_n = (x_j, x_m)$ with $x_j\neq x_m$.
Next, let 
\begin{equation}\label{E1a}
n_1 \in \min\{n: n\in\N\setminus R\}\;\text{with}\; q_{n_1}= (x_{j_1}, x_{m_1}), \;\text{and}
\end{equation}
\begin{equation}\label{E1b}
n^{\prime}_1 \in (\N\setminus R) \setminus \{n_1\}, \;\text{such that}\; q_{n^{\prime}_1}= (x_{m_1}, x_{j_1}).
\end{equation}
Having defined $(n_1, n^{\prime}_1)$,$\cdots$,  $(n_{k-1}, n^{\prime}_{k-1})$, let
\begin{equation}\label{E1c}
n_k \in \min\{n: n\in (\N\setminus R) \setminus \{n_1, n^{\prime}_1, \cdots, n_{k-1}, n^{\prime}_{k-1}\}\;\text{with}\; q_{n_k}= (x_{j_k}, x_{m_k}), \;\text{and}
\end{equation}
\begin{equation}\label{E1d}
n^{\prime}_k\in(\N\setminus R) \setminus \{n_1, n^{\prime}_1, \cdots, n_{k-1}, n^{\prime}_{k-1}, n_k\}, \;\text{such that}\; q_{n^{\prime}_k}= (x_{m_k}, x_{j_k}).
\end{equation}
Denote the disjoint subsets of $\N$ recursively defined in (\ref{E1c}) and (\ref{E1d}) as $A$ and $B$ respectively, i.e.,
\begin{equation}\label{E1e}
A:= \{n_k: k\in \N\}, \;\text{and}\; B:= \{n^{\prime}_k: k\in \N\}.
\end{equation}
The set $\N\setminus R$ has been partitioned in the sets $A$ and $B$ in the following manner.
First (minimum) and each subsequent element of $B$ lists the same pair of distinct elements but in reverse order as the first (minimum) and the corresponding subsequent  element of $A$.
Also let 
\begin{equation}\label{E1f}
\Gamma:= \{(n_k, n^{\prime}_k): n_k\in A, \;\text{and}\; n^{\prime}_k\in B, k\in \N\}.
\end{equation}
Set $\Gamma$ is a sequential listing of elements in $A\times B$.

\begin{proposition} \label{P1}
Let $\mathbf{T}$, $\mathbf{IR}$, $\mathbf{A}$, and $\mathbf{AS} \subseteq 2^\N$ be the coding sets of all transitive, irreflexive, asymmetric and antisymmetric  binary relations respectively.
Then there exists $N_f$, $N_{f^{\prime}}$, $N_g$, $N_{h} \in \varepsilon$ such that $N_f \subseteq \mathbf{T}$, $N_{f^{\prime}} \subseteq \mathbf{IR}$, $N_g \subseteq \mathbf{A}$ and $N_{h} \subseteq \mathbf{AS}$. 
In particular, $\mathbf{T}$, $\mathbf{IR}$, $\mathbf{A}$, and $\mathbf{AS}$ are not $\varepsilon$-small.
\end{proposition}

\begin{proof}

\begin{enumerate}
\item[]{\textbf{Transitivity}: 
It is sufficient to prove that there exists $N_f \in \varepsilon$ such that every $x \in N_f$ codes a transitive binary relation on $X$. 
Let $\dom(f) = A\cup B$ with $A, B\subset\N$ defined in (\ref{E1e}) and  let $f: \N \rightarrow 2$ be such that  for all $n \in \dom(f)$, $f(n)=0$. 
To show that every element in $N_f$ codes a transitive binary relation, pick $z \in N_f$ arbitrarily and let $\Re_z$ be the corresponding binary relation. 
Recall that transitivity is the following property: 
\[
\text{For all}\;  x, x^{\prime}, x^{\prime\prime}\in X, ((x, x^{\prime}) \in \Re_z \land (x^{\prime}, x^{\prime\prime}) \in \Re_z \Rightarrow (x, x^{\prime\prime}) \in \Re_z).
\]
 
As a consequence, $\Re_z$ trivially satisfies it because the left hand side of the implication never holds, since all pairs in $\Re_z$ are of the form $(x_n, x_n)$, unless $x=x^{\prime}=x^{\prime\prime}$ in which case the property trivially holds.
Therefore $N_f\subseteq \mathbf{T}$.}

\item[]{\textbf{Irreflexivity}: 
Let $f^{\prime}: \N \rightarrow 2$ be such that $\dom(f^{\prime})=R$ (where $R$ is defined in (\ref{E1})) and for all $n \in \dom(f^{\prime})$, $f(n)=0$. 
To show that every element in $N_f$ codes an irreflexive binary relation, pick $z \in N_{f^{\prime}}$ arbitrarily and let $\Re_z$ be the corresponding binary relation. 
Since $(x_{n_k}, x_{n_k})\notin \Re_z$ for all $k\in \N$ by construction $\Re_z$ is irreflexive.
Therefore $N_{f^{\prime}}\subseteq \mathbf{IR}$.}

\item[]{\textbf{Asymmetry}:
Define $g: \N \rightarrow 2$ so that $\dom(g):= B\cup R$ (where $R$ and $B$ are defined in (\ref{E1}) and (\ref{E1e}) respectively) and for all $i \in \dom(g)$, $g(i)=0$. 
Note that for every $z \in N_g$ we have three cases.
\begin{enumerate}
\item[(1)]{If $k \in B\subset \dom(g)$ then $z(k)=0$, which means the pair $(x_{m_k}, x_{j_k}) \notin \Re_z$.
So there is nothing to prove.}
\item[(2)]{If $k \in R \subset \dom(g)$ then $z(k)=0$, which means the pair $(x_{m_k}, x_{m_k}) \notin \Re_z$. 
Again there is nothing to prove.}
\item[(3)]{If $k \in A$ and $z(k)=1$, which means $(x_{j_k}, x_{m_k}) \in \Re_z$, we simply notice that $(x_{m_k}, x_{j_k}) \notin \Re_z$, by construction of $B$.
The case $k \in A$ and $z(k)=0$ is trivial as in case (1) above.}
\end{enumerate}
Hence $\Re_z$ satisfies asymmetry.
Therefore $N_g\subseteq \mathbf{A}$.}
\item[]{\textbf{Antisymmetry}:
Define $h: \N \rightarrow 2$ so that $\dom(h):= B$ (where $B$ is defined in (\ref{E1e})) and for all $i \in \dom(h)$, $h(i)=0$. 
Note that for every $z \in N_{h}$ we have three cases.
\begin{enumerate}
\item[(1)]{If $k \in B\subset \dom(h)$ then $z(k)=0$, which means the pair $(x_{m_k}, x_{j_k}) \notin \Re_z$.
So there is nothing to prove.}
\item[(2)]{If $k \in R$ and $z(k) =1$,  which means $(x_{m_k}, x_{m_k}) \in \Re_z$.
Again there is noting to prove. 
The case $k \in R$ and $z(k)=0$ is trivial as in case (1) above.}
\item[(3)]{If $k \in A$ and $z(k)=1$, which means $(x_{j_k}, x_{m_k}) \in \Re_z$, note that $(x_{m_k}, x_{j_k}) \notin \Re_z$, by construction of $B$.
The case $k \in A$ and $z(k)=0$ is trivial as in case (1) above.}
\end{enumerate}
Hence $\Re_z$ satisfies antisymmetry.
Therefore $N_{h}\subseteq \mathbf{AS}$.}
\end{enumerate}
\end{proof}
\begin{remark}
\emph{Proposition \ref{P1} distinguishes Ellentuck topology from Cantor topology. 
\citet{knoblauch2014} has shown that the binary relations satisfying basic properties are rare in Cantor topology.
In contrast, the transitive, asymmetric or antisymmetric  binary relations are not rare in Ellentuck topology.}
\end{remark}
 
\begin{proposition} \label{P2}
Let $\mathbf{C}$, $\mathbf{S}$, $\mathbf{R} \subseteq 2^\N$ be the coding sets of all complete, symmetric and reflexive binary relations, respectively. 
Then $\mathbf{C}$, $\mathbf{S}$, $\mathbf{R}$ are $\varepsilon$-small, but they contain open subsets in $\delta$, and so in particular, $\mathbf{C}$, $\mathbf{S}$, $\mathbf{R}$ are not $\delta$-small.
\end{proposition}

\begin{proof}
\begin{enumerate}
\item[(a)]{$\mathbf{C}$ is $\varepsilon$-small: 
Pick arbitrarily $N_f \in \varepsilon$, we aim to find $N_g \in \varepsilon$, $N_g \subseteq N_f$ such that $N_g \cap \mathbf{C} = \emptyset$. 
Let
\begin{equation}\label{P2E0}
\dom^{\prime}(f):= \left\{n \in \dom(f): f(n)=0\right\}. 
\end{equation}
Observe that $\dom^{\prime}(f)$ is an infinite set.
For $k \in \dom^{\prime}(f)$,
there are $j_k, m_k \in \N$ such that $q_k=(x_{j_k}, x_{m_k})$.
Also there is $k^{\prime}\in \N$ such that $q_{k^{\prime}}=(x_{m_k}, x_{j_k})$.
Pick $k\in\dom^{\prime}(f)$ such that  $k^{\prime} > k$.
There are two cases.
\begin{enumerate}
\item[(1)]{$k^{\prime} \in \dom(f)$: 
Then $N_f \cap \mathbf{C}= \emptyset$, for both $f(k)=f(k^{\prime})=0$ and so neither $(x_j, x_m)$ nor $(x_m, x_j)$ are in any binary relation coded by any $z \in N_f$.}
\item[(2)]{$k^{\prime} \notin \dom(f)$:
Then the partial function $g: \N \rightarrow 2$ with $\dom(g):= \dom(f) \cup \{k^{\prime}\}$ is defined as:
\begin{equation}\label{P2E1}
g(n) := \left\{ 
\begin{array}{ll}
f(n) & \text{ if $n \in \dom(f)$} \\
0 &  \text{if $n=k^{\prime}$}.
\end{array}
\right. 
\end{equation}
Note that $N_g$ is a well-defined subset in $\varepsilon$ and $N_g \subseteq N_f$. 
Moreover, since $g(k)=g(k^{\prime})=0$, it follows that neither $(x_j, x_m)$ nor $(x_m, x_j)$ are in any binary relation coded by any $z \in N_g$; which gives $N_g \cap \mathbf{C} = \emptyset$.}
\end{enumerate}}
\item[(b)]{$\mathbf{R}$ is $\varepsilon$-small:
Pick arbitrarily $N_f \in \varepsilon$, we aim to find $N_g \in \varepsilon$, $N_g \subseteq N_f$ such that $N_g \cap \mathbf{R} = \emptyset$. 
Let $R$ be as  in (\ref{E1}) and  $\dom^{\prime}(f)$ as in (\ref{P2E0}).
There are two cases.
\begin{enumerate}
\item[(1)]{$\dom^{\prime}(f)\cap R\neq \emptyset$:
Let $n_k\in \dom^{\prime}(f)\cap R$. 
Then $f(n_k) =0$ and therefore $N_f \cap \mathbf{R}= \emptyset$, since $(x_k, x_k)$ is not in any binary relation coded by any $z \in N_f$.}
\item[(2)]{$\dom^{\prime}(f)\cap R= \emptyset$: 
Choose $n_l\in R$ such that $n_l\notin \dom(f)$ and define the partial function $g: \N \rightarrow 2$ with $\dom(g):= \dom(f) \cup \{n_l\}$ as:
\begin{equation}\label{P2E2}
g(n) := \left\{ 
\begin{array}{ll}
f(n) & \;\text{if}\; n \in \dom(f) \\
0 &  \;\text{if}\; n=n_l.
\end{array}
\right. 
\end{equation}
Note that $N_g$ is a well-defined subset in $\varepsilon$ and $N_g \subseteq N_f$. 
Moreover, since $g(n_l)=0$, it follows that $(x_l, x_l)$ is not in any binary relation coded by any $z \in N_g$; which gives $N_g \cap \mathbf{R}= \emptyset$.}
\end{enumerate}}
\item[(c)]{$\mathbf{S}$ is $\varepsilon$-small: 
We leave this to the reader  as a simple exercise.}
\end{enumerate}

\begin{enumerate}
\item[(A)]{$\mathbf{C}$ is not $\delta$-small: 
We need to show that there is a set $N_f \in \delta$ such that $N_f \subseteq C$. 
First, we partition $\N$ into three sets $R$,  $A$ and  $B$ as in the proof of Proposition \ref{P1}.
Define $f: \N \rightarrow 2$ such that $\dom(f):= A \cup R$ and for all $i \in \dom(f)$, $f(i)=1$. 
Through the inductive construction, every pair $(x,y) \in X \times X$ has been considered.
Either $(x,y)$ or $(y,x)$ has been added to $A \cup R$. 
As a consequence, since every $z \in N_f$ takes value $1$ for all pairs coded in $A\cup R$, it follows that every $z \in N_f$ codes a complete binary relation.}
\item[(B)]{$\mathbf{R}$ is not $\delta$-small: 
This proof is included in case (A) above.}
\item[(C)]{$\mathbf{S}$ is not $\delta$-small: 
We leave this to the reader  as a simple exercise.}
\end{enumerate}
\end{proof}

Corollary \ref{C1} follows from Propositions \ref{P1} and \ref{P2}.
\begin{corollary}\label{C1}
Let $\mathbf{L}\subseteq 2^\N$ and $\mathbf{EQ}\subseteq 2^\N$ be the coding set of all linear binary relations and equivalence relations respectively. 
Then $\mathbf{L}$  and $\mathbf{EQ}$ are $\varepsilon$-small, but they contain open subsets in $\delta$, i.e., they are  not $\delta$-small.
\end{corollary}

Table \ref{T2} below summarizes the results. 
None of the basic properties of binary relation is rare in doughnut topology whereas all of them are rare in the Cantor topology.
Further, only transitive or asymmetric or antisymmetric or irreflexive binary relations are not rare in Ellentuck topology. 
If we put together these three observations we can say that transitivity, asymmetry, antisymmetry and irreflexivity are the less rare properties, since they are small with respect to both $\gamma$ and $\varepsilon$, whereas symmetry, reflexivity, completeness and linearity are small with respect to only one, namely $\gamma$. 
In particular Ellentuck collections plays an important role as it allows to make a distinction between the former four properties on the one side, and the latter four properties on the other.

\begin{table}[ht]
\centering
\begin{tabular}{|l|l | c |l |l| }
\hline 
\textbf{Property}&	$\gamma$-small&	$\varepsilon$-small&	 $\delta$-small  \\
\hline\hline
Transitivity&	yes &	no		&	no	\\\hline
Asymmetry& yes &		no		&	no 	\\\hline
Antisymmetry& yes &		no		&	no 	\\\hline
Irreflexivity&	yes &		no	&		no 	\\\hline
Reflexivity&	yes &		yes		&		no 	\\\hline
Symmetry&	yes &		yes		&		no 		\\\hline
Completeness&	yes &yes	 &		no   \\\hline
Linearity&	yes &yes	 &		no  \\\hline
\end{tabular}
\caption{Basic properties of binary relations}\label{T2}
\end{table}

\begin{remark}
\emph{Propositions \ref{P1} and \ref{P2} present an interesting  feature of Ellentuck topology with regard to reflexive and irreflexive binary relations.
For $|X|=n$, the numbers of all possible reflexive and irreflexive binary relations are same, i.e., $2^{n^2-n}$.
In the Ellentuck topology, the collection of reflexive binary relations is small whereas the collection of irreflexive binary relations does not turn out to be small.
In other words, the Ellentuck topology does not reflect the similar nature of the reflexive and irreflexive binary relations observed in the finite case, which seems to be counter-intuitive.}
\end{remark}

\section{Quasi-orders, partial orders and Ellentuck topology} \label{s4}
Proposition \ref{P1} shows that the collection of all transitive binary relations does not generate a small collection of codes in the Ellentuck topology whereas  the collection of all reflexive binary relations yields a small collection of codes in the Ellentuck topology as shown in Proposition \ref{P2}.
This leads us to examine the relative size of the collection of all quasi-orders and the collection of all partial orders in greater detail.
Existing results in the literature approach this question by taking the ratio of number of all quasi-orders (denoted by $\mathbf{Q}(n)$) and the number of all partial orders (denoted by $\mathbf{P}(n)$) where the set of alternatives $X$ contains finitely many elements and then determine the value when $n$ approaches infinity.
\citet{erne1974} shows that $\mathbf{P}(n)$ and $\mathbf{Q}(n)$ are asymptotically equinumerable, i.e.,  
\begin{equation}\label{S4E1}
\frac{\mathbf{Q}(n)}{\mathbf{P}(n)} \rightarrow 1\; \text{as}\; n\rightarrow \infty.
\end{equation}
\citet{klaska1997} proves following two results on the ratio of number of all transitive binary relations  (denoted by $\mathbf{T}(n)$) and the number of all  quasi-orders and partial orders respectively.
\begin{equation*}\label{S4E2}
\frac{\mathbf{T}(n)}{2^n\mathbf{P}(n)} \rightarrow 1\; \text{as}\; n\rightarrow \infty,\; \text{and}\; \frac{\mathbf{T}(n)}{2^n\mathbf{Q}(n)} \rightarrow 1\; \text{as}\; n\rightarrow \infty.
\end{equation*}
which in particular implies
\begin{equation}\label{S4E3}
\frac{\mathbf{P}(n)}{\mathbf{T}(n)} \rightarrow 0\; \text{as}\; n\rightarrow \infty,\; \text{and}\; \frac{\mathbf{Q}(n)}{\mathbf{T}(n)} \rightarrow 0\; \text{as}\; n\rightarrow \infty.
\end{equation}
Note that (\ref{S4E3}) means that both $\mathbf{P}(n)$ and $\mathbf{Q}(n)$ are asymptotically small compared to  $\mathbf{T}(n)$.

In this section we provide some similar results also in the countably infinite case.
Since we have assumed the set of alternatives $X$ to have countably infinite elements, the relative numerosity can be captured by the notion of relative topology. 
Let 
\[
\mathbf{P}:=\left\{x\in 2^{\N}: x\;\text{is a code of a partial order on}\; X\times X\right\}, 
\]
\[
\mathbf{Q}:=\left\{x\in 2^{\N}: x\;\text{is a code of a quasi-order on}\; X\times X\right\}\;\text{and}
\]
\[
\mathbf{T}:=\left\{x\in 2^{\N}: x\;\text{is a code of a transitive binary relation on}\; X\times X\right\}.
\]
Then consider the following induced topologies.
\[
\varepsilon_{\mathbf{T}}:= \left\{B\cap \mathbf{T}: B\in\varepsilon\right\},  \gamma_{\mathbf{Q}}:= \left\{B\cap \mathbf{Q}: B\in\gamma\right\} \;\text{and}\; \varepsilon_{\mathbf{Q}}:= \left\{B\cap \mathbf{Q}: B\in\varepsilon\right\}.
\]
The idea is to estimate the smallness of a set in the induced topologies instead of the whole space $2^\mathbb{N}$ in order to get an evaluation in term of \emph{relative numerosity}. 
We report three results on the size of the collection of codes for quasi-orders $\mathbf{Q}$ and partial orders $\mathbf{P}$ in the induced topologies $\varepsilon_{\mathbf{T}}$, $\gamma_{\mathbf{Q}}$ and $\varepsilon_{\mathbf{Q}}$, which in a sense show that the Ellentuck topology better extends than Cantor topology Klaska's results (\ref{S4E3}) to the infinite case.

\begin{proposition}\label{P3}
$\mathbf{Q}$ is $\varepsilon_{\mathbf{T}}$-small.
\end{proposition}

\begin{proof}
We pick $N_f\in \varepsilon_{\mathbf{T}}$ arbitrarily and show that there exists $N_g\subseteq N_f$ such that $N_g\in \varepsilon_{\mathbf{T}}$ and $N_g\cap {\mathbf{Q}}=\emptyset$.

\begin{enumerate}
\item[(1)]{We use $R\subset\N$ defined in (\ref{E1}).
If $R\subseteq \dom(f)$, then there is nothing to prove since 
\[
\forall z\in \mathbf{Q}, \forall m\in R, z(m)=1, 
\]
whereas
\[
\forall z\in N_f, \exists k\in\N,\; \forall m\geq k, m\in \dom(f), z(m)= f(m)=0.
\]
Thus $\mathbf{Q}\cap N_f=\emptyset$.}
\item[(2)]{If $R\not\subseteq \dom(f)$, i.e., there exists a $\bar{m}\in R$ such that $\bar{m}\notin\dom(f)$.
Define partial function $g:\N\rightarrow 2$ with $\dom(g):= \dom(f) \cup \{\bar{m}\}$ as:
\begin{equation}\label{P3E2}
g(n) := \left\{ 
\begin{array}{ll}
f(n) & \;\text{if}\; n \in \dom(f) \\
0 &  \;\text{if}\; n=\bar{m}.
\end{array}
\right. 
\end{equation}
Note that $N_g$ is a well-defined subset in $\varepsilon_{\mathbf{T}}$ and $N_g \subseteq N_f$. 
Since
\[
\forall m\in R, \forall z\in \mathbf{Q}, z(m)=1, \text{whereas} \; \forall z\in N_g,  z(\bar{m})= g(\bar{m})=0,
\]
it follows $\mathbf{Q}\cap N_g=\emptyset$.}
\end{enumerate}
\end{proof}

Next, we show that the collection of partial orders $\mathbf{P}$ is small in $\gamma_{\mathbf{Q}}$.
However, it is not small in $\varepsilon_{\mathbf{Q}}$.

\begin{proposition}\label{P4}
$\mathbf{P}$ is  $\gamma_{\mathbf{Q}}$-small.
\end{proposition}

\begin{proof}
We need to show that for every $N_f\in \gamma_{\mathbf{Q}}$ there exists $N_g\in \gamma_{\mathbf{Q}}$,   $N_g\subseteq N_f$  and $N_g\cap \mathbf{P}=\emptyset$.
Pick $N_f\in \gamma_{\mathbf{Q}}$ arbitrarily.
Since $\mathbf{P}$ is antisymmetric ($\mathbf{P}\supseteq \mathbf{AS}$), given any pair $\left(m, m^{\prime}\right)\in \Gamma$ (where $\Gamma$ is defined in (\ref{E1f})) we have that for all $z\in \mathbf{P}$,
\[
z(m)=1\Leftrightarrow z(m^{\prime})=0.
\]
Pick $\left(m, m^{\prime}\right)\in \Gamma$, such that $m, m^{\prime}>n$ for all $n\in \dom(f)$, which is always possible since $|\dom(f)|$ is finite in $\gamma$-topology.
Then define partial function $g:\N\rightarrow 2$ with $\dom(g):= \dom(f) \cup \{m, m^{\prime}\}$ as:
\[
g(n) := \left\{ 
\begin{array}{ll}
f(n) & \;\text{if}\; n \in \dom(f) \\
1 &  \;\text{if}\; n\in \{m, m^{\prime}\}.
\end{array}
\right. 
\]
Note that $N_g$ is a well-defined subset in $\gamma_\mathbf{Q}$ and $N_g \subseteq N_f$. 
Since
\[
\forall z\in \mathbf{P}, z(m)=1\Leftrightarrow z(m^{\prime})=0\;\text{and}
\] 
\[
\forall z\in N_g,  z(m)= g(m)= z(m^{\prime})= g(m^{\prime})=1,
\]
we get $\mathbf{P}\cap N_g=\emptyset$.

\end{proof}

\begin{proposition}\label{P5}
$\mathbf{P}$ is not $\varepsilon_{\mathbf{Q}}$-small.
\end{proposition}
\begin{proof}
It is sufficient to prove that there exists $N_f \in \varepsilon_{\mathbf{Q}}$ such that every $x \in N_f$ codes a partial order on $X$. 
We use the set $\Gamma$ (where $\Gamma$ is defined in (\ref{E1f}))  and define  $f: \N \rightarrow 2$ as follows.
For all $\left(m, m^{\prime}\right)\in \Gamma$, we set
\[
f(m)=0,\;\text{and}\; m^{\prime}\notin \dom(f).
\]
Thus $N_f\in \varepsilon$ and hence $N_f\cap \mathbf{Q}\in \varepsilon_{\mathbf{Q}}$.
For all $z\in N_f\cap \mathbf{Q}$ and all $(m,m') \in \Gamma$, $z(m)=0$. 
Hence $z$ is a code of antisymmetric quasi-order, i.e., $N_f\cap \mathbf{Q}\subseteq \mathbf{P}$.
\end{proof}

\section{Egalitarian binary relations} \label{s5}

In this section we consider  anonymity- a procedural equity principle;  strong equity - a consequentialist equity principle and the Pareto axiom- an efficiency principle.
The utility space $X$ we consider here is countable. 
Take for instance $X:= Y^N$, for $N \in \N$ and $Y$ is any countable subset of $[0,1]$.
One could choose $Y:= \left\{ \frac{1}{n}: n \in \N \right\}$ or $Y:=\QQ \cap [0,1]$ for example.

\begin{proposition} \label{P6}
Let $\mathbf{AN} \subseteq 2^\N$ be the coding set of all anonymous binary relations on $X$. 
Then $\mathbf{AN}$ consists of a $\gamma$- and $\varepsilon$-small set, but it contains an $N_f \in \delta$, and so it is not $\delta$-small.
\end{proposition}

\begin{proof}
First we prove that there is $N_f \in \delta$ such that $N_f  \subseteq \mathbf{AN}$. 
Let $\{ e_n: n \in \N \}$ enumerate all streams in $Y^N$. 
Let $z \in 2^\N$ be defined as: for all $s, t \in Y^N$ $((s,t) \in \Re_z \Leftrightarrow s \sim_a t)$. 
Note that both $\{n \in \N: z(n)=1\}$ and $\{n \in \N: z(n)=0 \}$ are infinite. 
Finally let $f: \N \rightarrow 2$ be the function such that $\dom(f):= \{n \in \N: z(n)=1\}$ (in particular, $\forall n \in \dom(f) (f(n)=1)$).
Then note that every $x \in 2^\N$ is a code for an anonymous binary relation if it satisfies $x(n)=1$ for every $n \in \dom(f)$. 
Hence, every $x \in N_f$ codes an anonymous binary relation, i.e.,  $N_f \subseteq \mathbf{AN}$.

Next we prove that $\mathbf{AN}$ is $\varepsilon$-small, and notice that a similar (and actually simpler argument) shows $\mathbf{AN}$ is $\gamma$-small as well.  
Fix arbitrarily an element $N_f \in \varepsilon$.
Let $k \in \N$ large enough so that for all $n \geq k$, if $n \in \dom(f)$ then $f(n)=0$. 
Now pick $m \geq k$ such that $z(m)=1$.
As in the proof of Proposition \ref{P2}, we then distinguish two cases. 
\begin{enumerate}
\item[(1)]{If $m \in \dom(f)$, then $N_f \cap \mathbf{AN}= \emptyset$.}
\item[(2)]{If $m \notin \dom(f)$, then define the partial function $g: \N \rightarrow 2$ with $\dom(g):= \dom(f) \cup \{ m \}$ as:
\[
g(i) := \left\{ 
\begin{array}{ll}
f(i) & \text{ if $i \in \dom(f)$} \\
0 &  \text{ if $i=m$}.
\end{array}
\right. 
\]
Note that $N_g \in \varepsilon$ and $N_g \subseteq N_f$. 
Moreover, since $g(m)=0$, it follows that for any binary relation coded by any $z \in N_g$ we can find a pair $(t, t^{\prime})$ such that $t \sim_a t^{\prime}$ but $(t, t^{\prime})$ is not in the binary relation coded by $z$; which gives $N_g \cap \mathbf{AN} = \emptyset$.}
\end{enumerate} 
\end{proof}

Similar result holds for Paretian binary relation.  

\begin{proposition} \label{P7}
Let $\mathbf{PA} \subseteq 2^\N$ be the coding set of all Paretian binary relations on $X$. 
Then $\mathbf{PA}$ consists of a $\gamma$- and $\varepsilon$-small set, but it contains an $N_f \in \delta$.
\end{proposition}

\begin{proof}
The proof is a similar argument as in Proposition \ref{P3}. 
First, we prove that there exists $N_f \in \delta$ such that $N_f \subseteq \mathbf{PA}$. 
Define $A(n)\subset\N$ and $B(n)\subset(\N\setminus R)$ (where $R$ is as defined  in \ref{E1}) recursively as follows:
\begin{itemize}
\item Start from $q_1=(x_{j_1}, x_{m_1})$. 
Let $k(1)\in \N$ be such that $q_{k(1)}=(x_{m_1}, x_{j_1})$.
If $(x_{j_1}, x_{m_1})$ satisfies the Pareto condition, then put $A(1):= \{1\}$ and $B(1):= \{k(1)\}$.
Otherwise let $A(1)= \emptyset$ and $B(1)= \emptyset$.
\item Assume $A(n-1)$ and $B(n-1)$ have been defined for $n\geq 2$ and pick $q_n= (x_{j_n}, x_{m_n})$; if $(x_{j_n}, x_{m_n})$ satisfies the Pareto condition, then put $A(n):= A(n-1) \cup  \{ n \}$ and $B(n):= B(n-1) \cup  \{ k(n) \}$; otherwise let  $A(n) = A(n-1)$ and $B(n):= B(n-1)$.
\end{itemize}
Finally put $A:= \bigcup_{n \in \N} A(n)$ and $B:= \bigcup_{n \in \N} B(n)$.
By construction, both $A$ and $B$ are infinite.
Set $A$ enumerates all pairs of alternatives $(x_i, x_l)$ such that $x_i<_p x_l$.
For each element in set $A$, set $B$ enumerates all the corresponding opposite pairs of alternatives $(x_l, x_i)$ such that $x_l<_p x_i$.
Since $R$ is infinite, the complement of $A\cup B$ is also infinite.
Then define the partial function $f: \N \rightarrow 2$ with $\dom(f):= A \cup B$ as: 
\[
f(n) := \left\{ 
\begin{array}{ll}
1 & \;\text{if}\; n \in A \\
0 & \;\text{if}\; n \in B.
\end{array}
\right. 
\]
We claim that every $z \in N_f$ codes a Paretian binary relation. 
In fact, if $n \in N$ is such that the pair $(x_{j_n}, x_{m_n})$ satisfies the Pareto axiom, then $n \in A$ and the related $k(n) \in B$ and therefore $z(n)=f(n)=1$ (which means $(x_{j_n}, x_{m_n}) \in \Re_z$) and $z(k(n))=f(k(n))=0$ (which means $(x_{m_n}, x_{j_n}) \notin \Re_z$).  

The proof to show that $\mathbf{PA}$ is $\varepsilon$-small follows the same line as for Proposition \ref{P3}.
\end{proof}

Similar result holds for binary relations satisfying strong equity. 
The proof works exactly following the same lines as for Proposition \ref{P4} and we leave it to the reader. 

\begin{proposition} \label{P8}
Let $\mathbf{SE} \subseteq 2^\N$ be the coding set of all strong equity binary relations on $X$. 
Then $\mathbf{SE}$ consists of a $\gamma$- and $\varepsilon$-small set, but it contains an $N_f \in \delta$.
\end{proposition}

Table \ref{T3} summarizes these results. 

\begin{table}[ht] 
\centering
\begin{tabular}{|l|l | c |l |l| }
\hline 
\textbf{Property}&	$\gamma$-small&	$\varepsilon$-small&	 $\delta$-small   \\
\hline\hline
Anonymity&	yes &		yes&		no 	\\\hline
Paretian&	yes &		yes	&		no 	\\\hline
Strong equity &	yes &		yes	&		no \\\hline
\end{tabular}
\caption{Equity and efficiency properties of preference relations}\label{T3}
\end{table}

\section{Concluding remarks}\label{s6}

In this paper, we have used the Ellentuck and doughnut topologies  (from the branch of descriptive set theory in the mathematical logic literature) to investigate the rarity of binary relations endowed with useful basic features (transitive, asymmetric, etc.).
Propositions \ref{P1} and \ref{P2} show that these binary relations are not rare in the finest (doughnut) topology.
The Ellentuck topology yields mixed results.
Transitive, asymmetric or antisymmetric binary relations are not rare whereas complete, reflexive or symmetric binary relations are rare in Ellentuck topology.
These results lead to a better understanding of the  distinct nature of Cantor topology compared to the Ellentuck and doughnut topologies.
The next results, Propositions \ref{P3}-\ref{P5}, show that Ellentuck topology better extends than Cantor topology the relative numerosity of partial and quasi-orders over transitive binary relations. 
Finally, Propositions \ref{P6}-\ref{P8} on the equitable or Paretian binary relations show that none of them are rare in doughnut topology.

In future research, we intend to study the pervasiveness or rarity of the binary relations endowed with desirable features (basic properties, equity or efficiency axioms) on the set of alternatives $X$ containing uncountably many elements using analytical tools from the generalized descriptive set theory.

\newpage
\bibliographystyle{plainnat}
\setlength{\bibsep}{0pt}
\small{\bibliography{APAnonymity}}
\end{document}